\documentclass[aps,pre,showpacs,
amssymb,twocolumn
]{revtex4-2}
\usepackage{graphicx}
\usepackage{amsmath}
\usepackage{amsfonts}
\usepackage{amssymb}
\usepackage{epsfig,libertine}
\usepackage[libertine]{newtxmath}
\usepackage{color}
\usepackage{dsfont, hyperref, xcolor}
\usepackage{comment}
\usepackage{enumerate}

\usepackage{amsthm} 

\newcommand{\abs}[1]{\left\vert#1\right\vert}

\newcommand{\Tr}[1]{\text{Tr}\left\{#1\right\}}
\newcommand{\ParTr}[2]{\text{Tr}_{#1}\left\{#2\right\}}
\newcommand{\bra}[1]{\langle#1\vert}
\newcommand{\ket}[1]{\vert#1\rangle}
\newcommand\braket[2]{\langle#1|#2\rangle}

\newtheorem{theorem}{Theorem}
\newtheorem{corollary}{Corollary}
\newtheorem{lemma}{Lemma}

\begin{document}

\title{Role of quantum correlations in daemonic expected utility}

\author{Gianluca~Francica and Luca Dell'Anna}
\address{Dipartimento di Fisica e Astronomia e Sezione INFN, Università di Padova, via Marzolo 8, 35131 Padova, Italy}

\date{\today}

\begin{abstract}
Fluctuations can challenge the possibility of improving work extraction from quantum correlations. This uncertainty in the work extraction process
can be addressed resorting to the expected utility hypothesis which can provide an optimal method for work extraction.
We study a bipartite quantum system and examine the role of quantum correlations in a daemonic work extraction performed by certain local operations and classical communication.
Specifically, we demonstrate and explain how, depending on the so-called absolute risk aversion, a non-neutral risk agent, influenced by fluctuations, views quantum correlations differently from a neutral risk agent who is affected solely by the average work.

\end{abstract}

\maketitle

\section{Introduction}

The role of correlations among the parties of a quantum system in thermodynamics has been recently investigated focusing on work extraction from finite systems~\cite{acin15,banik16,Alimuddin19,Francica17,mauro19,Touil22,Francicacorr22,Puliyil22,Shi22,Imai23}.
In particular, the daemonic ergotropy~\cite{Francica17}, which can be defined as the maximum average work locally extractable from a bipartite quantum system by performing certain local operations and classical communication, gives a gain in work extraction if there are quantum correlations.
Beyond its standard definition, generalized measurements and multipartite extensions have been also discussed~\cite{mauro19}. Among its different uses, e.g., in continuously-monitored open quantum batteries~\cite{Morrone23}, daemonic ergotropy also found  applications in order to investigate the role of indefinite causal order structures in thermodynamics.
In particular, when the communication is not classical and it is achieved with a process matrix~\cite{brukner12}, the role of indefinite causal order has been investigated with the help of daemonic ergotropy~\cite{Francicagames22}. Furthermore, by considering a quantum switch~\cite{Chiribella13},
the activation of states by applying maps in an indefinite causal order~\cite{kyrylo22,kyrylo23} and the effects of non-Markovianity~\cite{Cheong23}  have been also investigated.
Here, we wonder whether and when there are advantages in using the daemonic ergotropy protocol of Ref.~\cite{Francica17} if we take in account fluctuations by means of a utility function, i.e., by using the expected utility hypothesis, first formalized by von Neumann and Morgenstern within the theory of games and economic behaviour eighty years ago~\cite{vonNeumann}. 
Recently, expected utility has been also related to fluctuation theorems~\cite{Ducuara23,Francicauti23}. In particular, in Ref.~\cite{Francicauti23} we investigated a possible relation with an entropy coming from a fluctuation theorem, which in certain cases can be a guideline for making a choice.
Concerning the work fluctuations, when the initial state is not incoherent with respect to the energy basis, there may not be a probability distribution for the work done, as proven by a no-go theorem~\cite{Perarnau-Llobet17} due to the existence of quantum contextuality~\cite{Lostaglio18}. Thus, we can adopt the quasiprobability distribution of work introduced in Ref.~\cite{Francica22}, which is selected if some fundamental conditions need to be satisfied~\cite{Francica222}.
In work extraction from thermally isolated quantum system the optimal expected utility has been introduced in Ref.~\cite{Francica24} in order to optimize the work extraction by taking into account also the fluctuations, which can be dominant in finite systems.
Here, we study the gain achieved from quantum correlations when fluctuations are taken into account. After a brief introduction to some preliminary notions in Sec.~\ref{sec.preli}, we examine the expected utility for the work extraction protocol in Sec.~\ref{sec.daemouti}. In particular, we aim to clarify the role of quantum correlations when the fluctuations are taken in account with a utility function. Finally, in Sec.~\ref{sec.conclusions} we summarize and further discuss the obtained results.

\section{Preliminaries}\label{sec.preli}
We start our discussion by introducing some preliminary notions, which are the protocol of work extraction (see Sec.~\ref{sec.ergotropy}), some rudiments about quantum correlations, i.e., the quantum discord and the entanglement (see Sec.~\ref{sec.entanglement}), the  expected utility hypothesis (see Sec.~\ref{sec.utility}) and the quasiprobability (see Sec.~\ref{sec.quasiprobability}).
\subsection{Work extraction}\label{sec.ergotropy}
We consider a bipartite system having Hilbert space $\mathcal H = \mathcal H_S \otimes \mathcal H_A$, where $\mathcal H_S$ is the Hilbert space of a system $S$, where we can perform unitary transformations, and $\mathcal H_A$ is the one of a system (a so-called ancilla) $A$, where we can perform projective measurements. The two subsystems $S$ and $A$ are not interacting but they are prepared in a state $\rho_{SA}$ that can show correlations among the parties. A `daemonic' protocol can be realized through local operations and classical communication. In this section, we are interested to extract the optimal work locally from $S$ by neglecting the fluctuations and focusing only on the average extracted work, following Ref.~\cite{Francica17}.
We consider an Hilbert $\mathcal H_S$ with dimension $d_S$ and the Hamiltonian of the system $S$ is
\begin{equation}\label{eq.hami}
H_S=\sum_{k} \epsilon_k \ket{\epsilon_k}\bra{\epsilon_k}\,,
\end{equation}
with $\epsilon_{k}< \epsilon_{k+1}$. The reduced state of the system $S$ is given by
\begin{equation}
\rho_S = \ParTr{A}{\rho_{SA}} = \sum_k r_k \ket{r_k} \bra{r_k}\,,
\end{equation}
with $r_{k}\geq r_{k+1}$.
The system $S$ is thermally isolated, and an amount of average work is locally extracted by  cyclically changing some Hamiltonian parameters of $S$, so that at the end of the cycle the final Hamiltonian is equal to the initial one, and they are both equal to $H_S$ in Eq.~\eqref{eq.hami}. It results a unitary cycle with a local unitary time-evolution operator $U_S=\mathcal T e^{-i \int_0^\tau H_S(t)dt}$ for $S$, generated by the time-dependent Hamiltonian $H_S(t)$ such that $H_S(0)=H_S(\tau)=H_S$, where $t=0$ and $t=\tau$ are the initial and final time, and $\mathcal T$ is the time ordering operator. The final reduced state of $S$ is $U_S\rho_S U_S^\dagger$ and the average work is minus the change of average energy and reads
\begin{equation}\label{eq.workave}
W(\rho_S,U_S)=E(\rho_S)-E(U_S \rho_S U_S^\dagger)\,,
\end{equation}
where the average energy of the initial and final state is calculated with respect to the Hamiltonian $H_S$, and we have defined $E(\rho_S)=\Tr{H_S\rho_S}$. The optimal work extraction is achieved by maximizing the average work in Eq.~\eqref{eq.workave} over all the unitary cycles, i.e.,
\begin{equation}\label{eq.ergotropy def}
\mathcal E (\rho_S) = \max_{U_S} W(\rho_S,U_S) \geq 0\,.
\end{equation}
The optimal value $\mathcal E (\rho_S)$ is called ergotropy~\cite{Allahverdyan04}, and it is achieved by performing an optimal unitary cycle with unitary operator $U_S = \sum_{k=1}^d e^{i\phi_k} \ket{\epsilon_k}\bra{r_k}$, so that the ergotropy reads
\begin{equation}
\mathcal E (\rho_S) = \sum_{k,j} r_j\left( \abs{\braket{\epsilon_k}{r_j}}^2-\delta_{j,k}\right)\epsilon_k\,.
\end{equation}
We note that the ergotropy $\mathcal E(\rho_S)$ is zero if and only if the initial state $\rho_S$ is passive, i.e., commutates with the Hamiltonian, $[\rho_S,H_S]=0$, and the populations with respect to the energy basis are sorted in decreasing order, $r_k=\bra{\epsilon_k} \rho_S\ket{\epsilon_k}$. 

A larger amount of work can be extracted through the following daemonic protocol, which exploits the information about the state of $S$ that is obtained by performing projective measurements on $A$. These measurements are described by the set of projectors $\{\Pi^A_a=\ket{a}\bra{a}\}$ with $\braket{a}{a'}=\delta_{a,a'}$, where $a=1,\ldots,d_A$ and $d_A$ is the dimension of the Hilbert space $\mathcal H_A$.
By performing a projective measurement with projector $\Pi^A_a$ on the state of the ancilla $A$, the state of the system $S$ collapses into the state
\begin{equation}
\rho_{S|a} = \frac{\ParTr{A}{I^S\otimes \Pi^A_a \rho_{SA} I^S\otimes \Pi^A_a}}{p_a}\,,
\end{equation}
with probability $p_a=\Tr{I^S\otimes \Pi^A_a \rho_{SA}}$. Through classical communication, we can perform unitary cycles conditioned by the outcomes $a$ of the measurements, extracting the maximum average work $\mathcal E_{\{\Pi^A_a\}}(\rho_{SA})$ that is called daemonic ergotropy~\cite{Francica17} and reads
\begin{equation}
\mathcal E_{\{\Pi^A_a\}}(\rho_{SA}) = \sum_a p_a \mathcal E (\rho_{S|a})\,.
\end{equation}
We note that if the information gained from the measurements on $A$ is not exploited, the cycles are not conditioned by $a$ and the maximum average work extracted remains equal to the ergotropy $\mathcal E(\rho_S)$.
Thus, the maximum gain obtained by exploiting the information acquired reads
\begin{equation}\label{eq.gainergo}
\delta \mathcal E(\rho_{SA}) = \max_{\{\Pi^A_a\}} \mathcal E_{\{\Pi^A_a\}}(\rho_{SA}) - \mathcal E (\rho_S)\,,
\end{equation}
and it is related to the presence of quantum correlations~\cite{Francica17}.
\subsection{Quantum correlations}\label{sec.entanglement}
Quantum correlations will play a crucial role in our discussion. For our purposes, to identify them we define the set of classical-quantum states
\begin{equation}
\mathcal C_S = \left\{\rho_{SA} \,|\, \rho_{SA}= \sum_{k} p_{k} P^S_k \otimes \rho^A_k\right\}
\end{equation}
and the set of separable states
\begin{equation}
\mathcal S = \left\{\rho_{SA}\, |\, \rho_{SA}= \sum_{k} p_{k} \rho^S_k \otimes \rho^A_k\right\}\,,
\end{equation}
so that $\mathcal C_S \subseteq \mathcal S$, where $p_{k}\geq 0$ and sum to one, $P^S_k$'s are rank one projectors, $P^S_k P^S_j = \delta_{k,j}P^S_k$ and $\sum_{k} P^S_k = I^S$, and $\rho^S_k$ and $\rho^A_k$ are density matrices. We recall that all the separable states $\rho_{SA}\in \mathcal S$ can be prepared from a product state $\rho_S\otimes \rho_A$ by performing local operations and classical communication. On the other hand, states $\rho_{SA}\notin \mathcal S$ are called entangled (see, e.g., Ref.~\cite{Horodecki09} for a review). Correlations can be quantified through distance based measures (see, e.g., Ref.~\cite{Modi10}). For instance, the entanglement in a state $\rho_{SA}$ can be quantified through the relative entropy of entanglement~\cite{Vedral97}, $E_{re}(\rho_{SA})=\min_{\sigma_{SA}\in \mathcal S} S(\rho_{SA}||\sigma_{SA})$, where we have defined the quantum relative entropy $S(\rho||\eta)=\Tr{\rho(\log_2 \rho -\log_2 \eta)}$. 
Although a state $\rho_{SA}$ is separable, it can still show quantum features (quantum discord), if $\rho_{SA}\notin \mathcal C_S$. Originally, quantum discord has been introduced as the difference between the mutual information and the maximum one way classical information~\cite{Henderson01,Ollivier01}. A measure based on quantum relative entropy reads $D_{re}^{A|S}(\rho_{SA})=\min_{\chi^{c-q}_{SA}\in \mathcal C_S} S(\rho_{SA}||\chi^{c-q}_{SA})$. In particular, given a classical-quantum state $\chi^{c-q}_{SA}\in \mathcal C_S$ there exists a set  $\{P^S_k\}$ such that the state remains unperturbed if we perform these projective measurements, i.e., $\sum_k P^S_k\otimes I^A \chi^{c-q}_{SA}P^S_k\otimes I^A =\chi^{c-q}_{SA}$. 

\subsection{Expected utility hypothesis}\label{sec.utility}
Here, we aim to investigate the work extraction performed by an agent non-neutral to risk, which takes in account also the fluctuations and not only the average work.
For our purposes, we focus on an agent who must choose between two procedures that yield two different values of extracted work represented by the random variables $w_1$ and $w_2$.
To give an example, we consider an agent who must choose between extracting a certain work $w_{det}=50$ or flipping a coin and extracting a work $w_{head}=100$ if heads or nothing otherwise. If the agent is risk neutral, he is indifferent to the choice, since if he flips the coin he will extract the average work $w_{det}$. An agent non-neutral to risk will choose the certain work $w_{det}$ or to flip the coin depending on his risk aversion (e.g., an agent that is averse to risk, tends to choose the deterministic work extraction of the amount $w_{det}$, preferring situations with small fluctuations).
The risk aversion of the agent can be fully characterized by using a utility function $u(w)$, which quantifies the satisfaction gained from a choice, so that the agent will choose the procedure yielding the work $w_1$ instead of $w_2$ if
\begin{equation}\label{eq.exp utility theo}
\langle u(w_1) \rangle > \langle u(w_2) \rangle\,.
\end{equation}
It is easy to see that the inequality in Eq.~\eqref{eq.exp utility theo} remains unchanged if we perform an affine transformation on the utility function, i.e., the transformation $u(w) \mapsto a u(w) + b$, where $a$ is a positive variable. This means that the utility function is defined up to affine transformations, since two utility functions related by such transformation gives the same preference ordering given by Eq.~\eqref{eq.exp utility theo}. From Eq.~\eqref{eq.exp utility theo}, any linear utility function $u(w)=a w +b$, with $a>0$,  gives the condition $\langle w_1\rangle > \langle w_2 \rangle$, thus this utility function characterizes an agent neutral to risk. To characterize the risk aversion we can focus on a strictly increasing utility function that is concave. In our example, from Jensen's inequality, for arbitrary values of $w_{det}$ and $w_{head}$, the agent chooses the certain work $w_{det}$ instead to flip the coin if $w_{det}>w_{head}/2$. Thus, this suggests that the risk aversion is related to the concavity of the utility function, and we recall that in general it can be measured with the Arrow-Pratt coefficient of absolute risk aversion defined as
\begin{equation}\label{eq. RA}
r_A(w) = -\frac{u''(w)}{u'(w)}\,,
\end{equation}
which is non-negative for a utility function that is concave and strictly increasing. We note that, while $u''(w)$ quantifies the concavity, dividing by $u'(w)$ guarantees that the measure in Eq.~\eqref{eq. RA} is invariant under affine transformations. More details can be found, e.g., in Refs.~\cite{bookmicroeco,bookmicroeco2}.

\subsection{Quasiprobability}\label{sec.quasiprobability}
In general, we can get an extracted work $w$ with a quasiprobability distribution $p(w)$, such that $\int p(w) dw= 1$ and $p(w)$ is negative for some $w$. Of course, in this case results related to the risk aversion do not apply due to the negativity of $p(w)$. For instance, given a concave function $u(w)$, we get the Jensen's inequality $\langle u(w) \rangle \leq u(\langle w \rangle)$ if $p(w)\geq 0$, where the average is given by $\langle u(w) \rangle = \int u(w) p(w) dw$. However, the inequality can be not satisfied, i.e., we can get $\langle u(w) \rangle > u(\langle w \rangle)$, if $p(w)$ takes also negative values, suggesting that the implications of negativeness deserve a separate study. For our purposes, a utility function $u(w)$ defines an ordering of the quasiprobability distributions, i.e., given two quasiprobability distributions $p_1(w_1)$ and $p_2(w_2)$, we define $p_1(w_1)\succ p_2(w_2)$ if and only if $\langle u(w_1) \rangle > \langle u(w_2) \rangle$, i.e., Eq.~\eqref{eq.exp utility theo} holds, where we defined $\langle u(w_i) \rangle = \int u(w_i) p_i(w_i) dw_i$, with $i=1,2$. Defining an ordering is necessary to perform an optimization similar to that done in Eq.~\eqref{eq.ergotropy def}. Thus, given a state $\rho_S$, we consider the set $\mathcal A$ formed by the quasiprobability distributions of work corresponding to the different realizations of the time-evolution described by the unitary operators $U_S$. Then, from the ordering relation $p_1(w_1)\succ p_2(w_2)$, we can perform an optimization of the work extraction, i.e., we can find the sup of the set $\mathcal A$. In particular, if $u(w)=w$ we achieve Eq.~\eqref{eq.ergotropy def}, since $\langle w \rangle = W(\rho_S,U_S)$. However, $u(w)$ can be non-linear and thus higher moments are involved in the optimization.
In general, for a state $\rho_S$ and a time-evolution unitary operator $U_S$, there can be different quasiprobability representations of the work, forming a set $\mathcal C \subseteq \mathcal A$, which depends on $\rho_S$ and $U_S$. In this case the ordering of the processes can depend on the specific choice of the quasiprobability in $\mathcal C$, i.e., we can have $p_1(w_1)\succ p_2(w_2)$ and $p'_2(w_2)\succ p'_1(w_1)$ for some $p_i(w_i),p'_i(w_i)\in \mathcal C_i$, although all the quasiprobability distributions in $\mathcal C_i$ represent the same process, with $i=1,2$. Thus, we can optimize the work extraction by fixing a quasiprobability in $\mathcal C$ or by defining an ordering relation $\mathcal C_1 \succ_{\mathcal C} \mathcal C_2$. For instance, we can define $\mathcal C_1 \succ_{\mathcal C} \mathcal C_2$ if and only if $\forall p_1(w_1)\in \mathcal C_1$, $\exists p_2(w_2) \in \mathcal C_2$ such that $p_1(w_1)\succ p_2(w_2)$. In this case, $\mathcal C_1 \succ_{\mathcal C} \mathcal C_2$ and $\mathcal C_2 \succ_{\mathcal C} \mathcal C_3$ imply $\mathcal C_1 \succ_{\mathcal C} \mathcal C_3$, and thus we get a valid ordering. In particular, this definition of $\mathcal C_1 \succ_{\mathcal C} \mathcal C_2$ is equivalent to $\min_{p_1(w_1)\in \mathcal C_1}\langle u(w_1) \rangle > \min_{p_2(w_2)\in \mathcal C_2}\langle u(w_2) \rangle$.
\begin{proof}
To show it, we note that if $\mathcal C_1 \succ_{\mathcal C} \mathcal C_2$  we get $\langle u(w_1) \rangle > \min_{p_2(w_2)\in \mathcal C_2}\langle u(w_2) \rangle$ $\forall p_1(w_1)\in \mathcal C_1$, then $\min_{p_1(w_1)\in \mathcal C_1}\langle u(w_1) \rangle > \min_{p_2(w_2)\in \mathcal C_2}\langle u(w_2) \rangle$. On the other hand, if $\min_{p_1(w_1)\in \mathcal C_1}\langle u(w_1) \rangle > \min_{p_2(w_2)\in \mathcal C_2}\langle u(w_2) \rangle$,  $\forall p_1(w_1) \in \mathcal C_1$ we get $\langle u(w_1)\rangle> \min_{p_2(w_2)\in \mathcal C_2}\langle u(w_2) \rangle$, then $\exists p^*_2(w_2)\in \mathcal C_2$ giving $\int u(w_2)p_2^*(w_2)dw_2=\min_{p_2(w_2)\in \mathcal C_2}\langle u(w_2) \rangle$, from which  $\mathcal C_1 \succ_{\mathcal C} \mathcal C_2$.
\end{proof}
Thus, by using this definition, we can optimize the work extraction by searching the maximum of $\min_{p(w)\in \mathcal C}\langle u(w)\rangle$ over the unitary operators $U_S$.

\section{Daemonic expected utility}\label{sec.daemouti}
Given a utility function $u(w)$, we can define the optimal expected utility of an arbitrary state $\rho_S$ as~\cite{Francica24}
\begin{equation}\label{eq.exputiopti}
\mathcal U(\rho_S) = \max_{U_S} \langle u(w) \rangle\,,
\end{equation}
where the average is calculated as
\begin{equation}\label{eq.aveu}
\langle u(w) \rangle = \int u(w) p_q(w,\rho_S,U_S) dw
\end{equation}
and $p_q(w,\rho_S,U_S)$ is the quasiprobability distribution of work defined as~\cite{Francica22,Francica222}
\begin{eqnarray}
\nonumber p_q(w,\rho_S,U_S)&=&\sum_{k,j,i} \text{Re}\bra{\epsilon_i}\rho_S\ket{\epsilon_j}\bra{\epsilon_j}U_S^\dagger \ket{\epsilon_k}\bra{\epsilon_k}U_S \ket{\epsilon_i}\\
 && \times\delta(w-q\epsilon_i-(1-q)\epsilon_j+\epsilon_k)\,.
\end{eqnarray}
We note that for a risk neutral agent, so that $u(w)=w$, the optimal expected utility is equal to the ergotropy, $\mathcal U(\rho_S)=\mathcal E(\rho_S)$.
In particular, we will focus on the symmetric quasiprobability representation with $q=1/2$, giving a minimum of $\mathcal U(\rho_S)$ in the function of $q$ (at least in the case of an exponential utility~\cite{Francica24}). We recall that $p_q(w,\rho_S,U_S)$ reduces to the quasiprobability distribution of Ref.~\cite{Allahverdyan14} for $q=0,1$ and the one of Ref.~\cite{Solinas15} for $q=1/2$.
Thus, we consider a daemonic protocol for the work extraction, with certain optimal unitary cycles not necessarily equal to the daemonic ergotropy ones.  We denote with $U_{S|a}$ the unitary operators of the optimal unitary cycles conditioned by the outcomes $a$ of the measurements, so that the final reduced state of $S$ is $U_{S|a}\rho_{S|a}U^\dagger_{S|a}$ with probability $p_a$. This daemonic protocol gives an optimal expected utility which is equal to the `daemonic' expected utility defined as
\begin{equation}
\mathcal U_{\{\Pi^A_a\}} ( \rho_{SA}) = \sum_a p_a \mathcal U ( \rho_{S|a})\,,
\end{equation}
where explicitly $\mathcal U ( \rho_{S|a}) = \int u(w) p_q(w,\rho_{S|a},U_{S|a})dw$. 
We note that, for any convex combination $\rho_S=\sum_a p_a \rho_a$ of density matrices $\rho_a$,  we get $p_q(w,\rho_S,U_S) = \sum_a p_a p_q(w,\rho_a,U_S)$ from which it is easy to see that $\mathcal U(\rho_S)\leq \sum_a p_a \mathcal U(\rho_a)$. Thus, by noting that $\rho_S=\sum_a p_a \rho_{S|a}$, we get $\mathcal U_{\{\Pi^A_a\}} ( \rho_{SA})\geq \mathcal U(\rho_S)$.
We define the maximum gain
\begin{equation}
\delta \mathcal U (\rho_{SA})= \max_{\{\Pi^A_a\}}\mathcal U_{\{\Pi^A_a\}} ( \rho_{SA}) -\mathcal U ( \rho_S)\geq 0\,.
\end{equation}
Thus, the agent prefers the optimal daemonic protocol instead of a local unitary cycle in order to extract work from $S$ if the gain is positive, $\delta \mathcal U (\rho_{SA})>0$, otherwise the agent is indifferent to the choice. The situation is schematically represented in Fig.~\ref{fig:scheme}.
\begin{figure}
[t!]
\centering
\includegraphics[width=0.65\columnwidth]{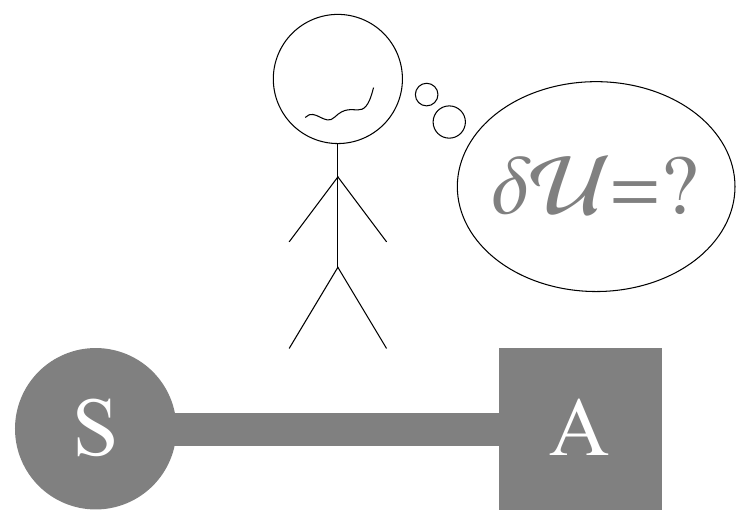}
\caption{The system is made of two parties, $S$ and $A$, represented by a circle and a square, respectively. The two parties are connected by a line representing the initial correlations. Local measurements and unitary operations are performed on the square and the circle, respectively. An agent must choose between extracting work from $S$ using the daemonic protocol or conventionally using a local unitary cycle without performing measurements on $A$ or communicating the outcomes $a$.
}
\label{fig:scheme}
\end{figure}
\subsection{Constant absolute risk aversion}
We start to investigate the case of a constant absolute risk aversion $r_A(w)=\text{const}$. In particular, we focus on the exponential utility
\begin{equation}\label{eq.uteqxpo}
u(w) = \frac{1}{r}(1-e^{-rw})
\end{equation}
for $r\neq 0$, and $u(w)=w$ for $r=0$, which is a strictly increasing function having absolute risk aversion $r_A(w)=r$.
To perform our calculations for $q=1/2$, it is useful to consider the spectral decompositions $e^{-rH_S/2}\rho_S e^{-rH_S/2}=\sum_k u_k \ket{u_k}\bra{u_k}$ with $u_k\geq u_{k+1}$ and $e^{-rH_S/2}\rho_{S|a} e^{-rH_S/2}=\sum_k u^a_k \ket{u^a_k}\bra{u^a_k}$ with $u^a_k\geq u^a_{k+1}$ for all $a$.
 Thus, for $q=1/2$, for any state $\rho_S$, from Eq.~\eqref{eq.exputiopti} we get~\cite{Francica24}
\begin{equation}\label{eq.exputiexpo}
\mathcal U ( \rho_S) = \frac{1}{r}\left(1-\sum_k u_k e^{r \epsilon_k}\right)\,,
\end{equation}
from which the daemonic expected utility  trivially reads
\begin{equation}
\mathcal U_{\{\Pi^A_a\}} ( \rho_{SA})  = \sum_a p_a \frac{1}{r}\left(1-\sum_k u^a_k e^{r \epsilon_k}\right)\,.
\end{equation}
With the aim to study the states $\rho_{SA}$ such that $\delta \mathcal U (\rho_{SA})=0$ in the case of a constant absolute risk aversion, we introduce some lemmas.
By considering the identity
\begin{equation}
u_k = \sum_a p_a \sum_j u^a_j \abs{\braket{u_k}{u^a_j}}^2\,,
\end{equation}
coming from $\rho_S =\sum_a p_a \rho_{S|a}$, we get that:
\begin{lemma}\label{lemma1}
For a given set $\{\Pi^A_a\}$, the daemonic gain $\mathcal U_{\{\Pi^A_a\}} ( \rho_{SA}) -\mathcal U ( \rho_S)$ can be expressed as
\begin{equation}
\mathcal U_{\{\Pi^A_a\}} ( \rho_{SA}) -\mathcal U ( \rho_S)= \sum_a p_a \tilde{\mathcal E}(\tilde \rho_{S|a})\geq 0\,,
\end{equation}
where $\tilde{\mathcal E}(\tilde \rho_{S|a})$ is the ergotropy of the non-normalized state $\tilde \rho_{S|a} = e^{-rH_S/2}\rho_{S|a}e^{-rH_S/2}$ with respect to the Hamiltonian $\tilde H_S = \sum_k y_k \ket{u_k}\bra{u_k}$ with energies $y_k=\frac{e^{r\epsilon_k}}{r}$ so that $y_k<y_{k+1}$, which explicitly reads
\begin{equation}
\tilde{\mathcal E}(\tilde \rho_{S|a})=\sum_{k,j} u^a_j\left( \abs{\braket{u_k}{u^a_j}}^2-\delta_{j,k}\right)y_k\,.
\end{equation}
\end{lemma}
Thus, $\mathcal U_{\{\Pi^A_a\}} ( \rho_{SA}) = \mathcal U ( \rho_S)$ if and only if all the ergotropies $\tilde{\mathcal E}(\tilde \rho_{S|a})$ are zero, i.e., for all the outcomes $a$, $\tilde \rho_{S|a}$ is passive with respect to the Hamiltonian $\tilde H_S$. This means that, for all the outcomes $a$, $[\tilde \rho_{S|a},\tilde H_S]=0$, i.e., all the conditional non-normalized states $\tilde \rho_{S|a}$ are diagonal with respect to the basis $\{\ket{u_k}\}$, and $\bra{u_k} \tilde \rho_{S|a}\ket{u_k}$ are in decreasing order, i.e.,$\bra{u_k} \tilde \rho_{S|a}\ket{u_k}\geq \bra{u_{k+1}} \tilde \rho_{S|a}\ket{u_{k+1}}$. Furthermore, we get that:
\begin{lemma}\label{lemma2}
Given a state $\rho_{SA}$, if all the conditional states $\rho_{S|a}$ are diagonal with respect to the same basis $\{\ket{r_k}\}$ for any set $\{\Pi^A_a\}$, then
\begin{equation}\label{eq.statelemma}
\rho_{SA} = \sum_k \ket{r_k}\bra{r_k}\otimes C^A_k\,,
\end{equation}
where $C^A_k$ are positive semidefinite matrices.
\end{lemma}
We note that the state in Eq.~\eqref{eq.statelemma} is a classical-quantum state, $\rho_{SA}\in \mathcal C_S$.
\begin{proof}
To prove Eq.~\eqref{eq.statelemma}, we note that an arbitrary state $\rho_{SA}$ can be written as
\begin{equation}\label{eq.arbilemma1}
\rho_{SA} = \sum_{a,a',k,k'} C^{a a'}_{kk'} \ket{r_k}\bra{r_{k'}} \otimes \ket{a} \bra{a'}\,,
\end{equation}
where $\{\ket{a}\}$ is a basis of $\mathcal H_A$.
By assumption, we have that $\rho_{S|a} = \sum_k r^a_k \ket{r_k} \bra{r_k}$ for all $a$, thus $C^{aa}_{kk'}=p_a r^a_k \delta_{kk'}$, where $p_a$ is the probability corresponding to the outcome $a$.
We expand the sum in Eq.~\eqref{eq.arbilemma1} getting
\begin{eqnarray}\label{eq.statedimo}
\nonumber \rho_{SA} &=& \sum_{a,k} p_a r^a_k \ket{r_k}\bra{r_k} \otimes \ket{a} \bra{a}+ \sum_{a\neq a',k} C^{aa'}_{kk} \ket{r_k}\bra{r_k} \otimes \ket{a}\bra{a'}\\
&&+\sum_{a>a',k\neq k'} \ket{r_k}\bra{r_{k'}} \otimes (C^{aa'}_{kk'}\ket{a}\bra{a'}+C^{a'a}_{kk'}\ket{a'}\bra{a})\,.
\end{eqnarray}
Thus, if the last term is zero, we get Eq.~\eqref{eq.statelemma} with the matrices $C^A_k$ having entries $C^{aa'}_{kk}$ with respect to the basis $\{\ket{a}\}$. To show that it is zero, we note that if there exist at least $ \bar a $ and $\bar a'\neq a$ such that $C^{\bar a\bar a'}_{kk'}\neq0$ and $C^{\bar a\bar a'}_{kk'}\neq 0$ for some $k\neq k'$ and the last term is non-zero, we can perform measurements with respect to some basis different from $\{\ket{a}\}$, e.g., containing the states $(\ket{\bar a} \pm \ket{\bar a'})/\sqrt{2}$, getting $\rho_{S|\bar a}$ non-diagonal with respect to the basis $\{\ket{r_k}\}$, which is absurd since $\rho_{S|a}$ is diagonal for any set $\{\Pi_a\}$, which completes the proof.
\end{proof}
In general we define a dephasing map $\Delta$ as $\Delta(X) = \sum_k \Pi_k X \Pi_k$, where $\Pi_k=\ket{k}\bra{k}$ such that $\Pi_k \Pi_j = \delta_{k,j}\Pi_k$ and $\sum_k \Pi_k=I$. By recalling that $A\geq B$ means that the matrix $A-B$ is positive semidefinite, we get the following lemma:
\begin{lemma}\label{lemma3}
Given two matrices $A$ and $B$, $\Delta(A)\geq \Delta(B)$ for any dephasing map $\Delta$ if and only if $A\geq B$.
\end{lemma}
\begin{proof}
To prove it, we note that the dephasing map is a linear map, thus $\Delta(A)\geq \Delta(B)$ is equivalent to $\Delta(A-B)\geq 0$. Thus, $\Delta(A)\geq \Delta(B)$ for any dephasing map $\Delta$ if and only if the diagonal elements of the matrix $A-B$ are non-negative with respect to any basis. This is equivalent to $\bra{\psi} A-B \ket{\psi}\geq 0$ for any state $\ket{\psi}$, i.e., $A-B$ is positive semidefinite, $A- B\geq 0$, which completes the proof.
\end{proof}
Thus, with the help of these three lemmas, we get the theorem:
\begin{theorem}\label{theorem}
We have a zero maximum gain $\delta \mathcal U (\rho_{SA}) = 0 $ if and only if $\rho_{SA}\in \mathcal S$ and has the form
\begin{equation}\label{eq.statetheo}
\rho_{SA} = \sum_{k}e^{rH_S/2}\ket{u_k}\bra{u_k}e^{rH_S/2} \otimes C^A_k\,,
\end{equation}
where the positive semidefinite matrices $C^A_k\geq 0$ are such that $C^A_k\geq C^{A}_{k+1}$ for all $k=1,\ldots,d_S-1$.
\end{theorem}
\begin{proof}
From Lemma~\ref{lemma1}, we get that $\delta \mathcal U (\rho_{SA}) = 0 $ if and only if $\tilde\rho_{S|a}$ are passive with respect to the Hamiltonian $\tilde H_S$ for all the outcomes $a$ for any set $\{\Pi^A_a\}$. Thus, by considering $\tilde \rho_{SA}= e^{-rH_S/2}\otimes I^A\rho_{SA}e^{-rH_S/2}\otimes I^A$, since $\tilde\rho_{S|a}$ is diagonal with respect to the basis $\{\ket{u_k}\}$ for any set $\{\Pi^A_a\}$, from Lemma~\ref{lemma2} we get $\tilde \rho_{SA} = \sum_{k}\ket{u_k}\bra{u_k} \otimes C^A_k$, from which it results Eq.~\eqref{eq.statetheo}. Furthermore, $\bra{u_k}\tilde\rho_{S|a}\ket{u_k}$ are in decreasing order for any set $\{\Pi^A_a\}$, by noting that $\bra{u_k}\tilde\rho_{S|a}\ket{u_k} = \bra{a}C^A_k\ket{a}$ where we have considered $\Pi^A_a = \ket{a}\bra{a}$, this is equivalent to $\Delta(C^A_{k})\geq \Delta(C^A_{k+1})$ for any dephasing map $\Delta$. Thus, from Lemma~\ref{lemma3} we get $C^A_k\geq C^{A}_{k+1}$.
\end{proof}
We note that the set of the states $\rho_{SA}$ with zero gain $\delta \mathcal U(\rho_{SA})=0$ has measure zero since these states are obtained from classical-quantum states $\chi^{c-q}_{SA}$ by performing the transformation $\rho_{SA}= e^{rH_S/2}\otimes I^A \chi^{c-q}_{SA} e^{rH_S/2}\otimes I^A/\Tr{\chi^{c-q}_{SA}e^{r H_S}\otimes I^A}$, which form a set of measure zero~\cite{alessandrodiscord}. Thus, almost all states $\rho_{SA}$ give a positive gain $\delta \mathcal U(\rho_{SA})>0$.
Furthermore, from Theorem~\ref{theorem} it follows that:
\begin{corollary}\label{corollary}
We get $r=0$ or $[\rho_S,H_S]=0$ if and only if $\delta\mathcal U(\rho_{SA})=0 \Rightarrow \rho_{SA}\in \mathcal C_S$.
\end{corollary}
\begin{proof}
For $r=0$ or if $[\rho_S,H_S]=0$, the states $e^{rH_S/2}\ket{u_k}$ are mutually orthogonal so that $\rho_{SA}$ in Eq.~\eqref{eq.statetheo} is a classical-quantum state, $\rho_{SA}\in \mathcal C_S$. Vice versa, if $\rho_{SA}$ in Eq.~\eqref{eq.statetheo} is a classical-quantum state, $e^{rH_S}$ is diagonal with respect to the basis $\{\ket{u_k}\}$, from which we get $r=0$ or $[\rho_S,H_S]=0$.
\end{proof}
In particular, for $r=0$ we get the ergotropy gain $\delta \mathcal U(\rho_{SA}) = \delta \mathcal E(\rho_{SA})$ in Eq.~\eqref{eq.gainergo}. This shows how the presence of initial quantum coherence with respect to the energy basis, so that $[\rho_S,H_S]\neq 0$, can make $\delta \mathcal U(\rho_{SA})=0$ for separable states $\rho_{SA}\notin \mathcal C_S$ with $\delta \mathcal E(\rho_{SA})>0$.
In order to illustrate this result with a simple example, we consider $d_S=2$, thus the state and the Hamiltonian
\begin{equation}\label{eq.rhoeHdS2}
\rho_S = \left(
           \begin{array}{cc}
             p & c \\
             c^* & 1-p \\
           \end{array}
         \right)\,,\quad
H_S = \left(
           \begin{array}{cc}
             \epsilon_1 & 0 \\
             0 & \epsilon_2 \\
           \end{array}
         \right)\,,
\end{equation}
where $\abs{c}\leq \sqrt{p(1-p)}$ so that $\rho_S\geq 0$.
For simplicity we consider $p=e^{r \epsilon_1}/Z_r$ and $c >0$, where $Z_r=e^{r \epsilon_1}+e^{r \epsilon_2}$. In this case, it is easy to see that $\rho_S$ can be obtained as $\rho_S= \ParTr{A}{\rho_{SA}}$ with
\begin{equation}\label{eq.exampleexpo}
\rho_{SA} = \sum_{\sigma=\pm}e^{r H_S/2} \ket{\sigma}\bra{\sigma}e^{r H_S/2} \otimes C^A_\sigma\,,
\end{equation}
where we have defined $\ket{\pm }=(\ket{\epsilon_1}\pm\ket{\epsilon_2})/\sqrt{2}$, and the operators $C^A_\pm$ are diagonal in the same basis $\{\ket{a}\}$, showing diagonal elements $C^{aa}_\pm$ such that $C^{00}_+=C^{00}_-\geq0$ and $C^{11}_+\geq C^{11}_-=0$. The equation $\rho_S= \ParTr{A}{\rho_{SA}}$ is satisfied if  $C^{00}_+ =1/Z_r-C^{11}_+/2$ and $C^{11}_+ = 2 c e^{-r(\epsilon_1+\epsilon_2)/2}$. By performing measurements with projectors $\{\Pi^A_a=\ket{a}\bra{a}\}$, we get $\rho_{S|0}=e^{r H_S}/Z_r$ and $\rho_{S|1}=2e^{r H_S/2} \ket{+}\bra{+}e^{r H_S/2}/Z_r$ with probabilities $p_0=1-p_1$ and $p_1=2c\cosh(r(\epsilon_2-\epsilon_1)/2)$.
For an absolute risk aversion $r_A=r$, the optimal expected utility $\mathcal U (\rho_S)$ is given by Eq.~\eqref{eq.exputiexpo}. For $r_A=r$, we get $u_{1,2}=1/Z_r\pm c e^{-r(\epsilon_1+\epsilon_2)/2}$, from which we get $\mathcal U(\rho_S)= \frac{2c}{r}\sinh(r(\epsilon_2-\epsilon_1)/2)$. Furthermore, we get $\mathcal U(\rho_{S|0})=0$ and $\mathcal U (\rho_{S|1})=\frac{1}{r}\tanh(r(\epsilon_2-\epsilon_1)/2)$, from which $\mathcal U_{\{\Pi^A_a\}}(\rho_{SA})=\sum_a p_a \mathcal U (\rho_{S|a}) = \mathcal U(\rho_S)$, and there is no gain in perfect agreement with Theorem~\ref{theorem}, since the separable state in Eq.~\eqref{eq.exampleexpo} is of the form of Eq.~\eqref{eq.statetheo}. Let us focus on a constant $r_A\neq r$, where $r$ is the parameter of the state in Eq.~\eqref{eq.exampleexpo}. If there is no initial quantum coherence, then $[\rho_S,H_S]=0$ and $c=0$, from which $p_1=0$. In this case we get $\rho_S = \rho_{S|0}$, and thus there is no gain although $r_A\neq r$. In particular, if $c=0$, then Eq.~\eqref{eq.exampleexpo} gives the product state $\rho_{SA}\propto\rho_S \otimes C^A_+$, since $C^A_-=C^A_+$, and thus $\delta \mathcal U(\rho_{SA})=0$, because any selective measurement done always gives a conditional state $\rho_{S|a}=\rho_S$.
Since $\rho_{SA}$ is a product state, $\rho_{SA}\in \mathcal C_S$ and the result is in perfect agreement with Corollary~\ref{corollary}.
In general, for $r_A=r'$, we get $\mathcal U(\rho_{S|0}) = \frac{4e^{r(\epsilon_2+\epsilon_1)/2}}{r' Z_r} \sinh(r'(\epsilon_2-\epsilon_1)/2)\sinh((r-r')(\epsilon_2-\epsilon_1)/2)$ if $r'<r$, $\mathcal U(\rho_{S|0}) = 0$ otherwise, and $\mathcal U(\rho_{S|1})= \frac{e^{r \epsilon_2}-e^{(r-r')\epsilon_2+r'\epsilon_1}}{r' Z_r}$, thus by calculating the eigenvalues $u_k$ of $e^{-r' H_S/2} \rho_S e^{-r' H_S/2}$, we get $\mathcal U(\rho_S)$ as Eq.~\eqref{eq.exputiopti}, and it results that $\mathcal U_{\{\Pi^A_a\}}(\rho_{SA})> \mathcal U(\rho_S)$ for $c>0$ and $r\neq r'$.
In detail, the daemonic protocol can be realized with $U_{S|0}=\ket{\epsilon_1}\bra{\epsilon_2}+\ket{\epsilon_2}\bra{\epsilon_1}$ if $r'<r$, $U_{S|0}=I^S$ otherwise, and $U_{S|1}$ defined such that $U_{S|1} \tilde\rho_{S|1}U^\dagger_{S|1}\propto\ket{\epsilon_1}\bra{\epsilon_1}$.
This illustrates how, while $\delta \mathcal U(\rho_{SA})=0$ when $r_A = r$, the presence of initial quantum coherence ($c\neq0$) can give a nonzero gain $\delta \mathcal U(\rho_{SA})>0$ when $r_A\neq r$. Then, for the state in Eq.~\eqref{eq.exampleexpo}, we have $\delta \mathcal U(\rho_{SA})=0$ for an absolute risk aversion $r_A=r$, but for $c>0$ we  have $\delta \mathcal U(\rho_{SA})>0$ for a constant absolute risk aversion $r_A\neq r$, e.g., for $r_A=0$, $\delta \mathcal E(\rho_{SA})>0$ if $r\neq 0$. In particular, $\delta \mathcal E(\rho_{SA})>0$ for the state in Eq.~\eqref{eq.exampleexpo} since $\rho_{SA}\notin \mathcal C_S$.
In general, for $d_S=d_A=2$, as shown in Ref.~\cite{Francica17}, by considering $\delta \mathcal E =\delta \mathcal E (\rho_{SA})$ in unit of $\epsilon_2-\epsilon_1$, we get $D^{A|S}\leq h(1-\delta \mathcal E/2)$, where the function $h(x)$ reads $h(x)=-x \log_2 x-(1-x)\log_2(1-x)$ if $D^{A|S}=D^{A|S}(\rho_{SA})$ is the quantum discord defined in Ref.~\cite{Ollivier01}.
Thus, the presence of quantum correlations gives a nonzero lower bound for the gain, so that $\delta \mathcal E\geq 2 - 2 h^{-1}(D^{A|S})$. Similarly, for $r_A=r\neq 0$, for randomly generated states $\rho_{SA}$, we will still get a lower bound for $\delta \mathcal U=\delta \mathcal U(\rho_{SA})$, so that $\delta \mathcal U\geq 2 - 2 h_r^{-1}(D^{A|S})$ for a certain function $h_r(x)$, since the set of the states with $D^{A|S}\neq 0$ and $\delta \mathcal U=0$ has measure zero.

\subsection{Non-constant absolute risk aversion}
In general, the agent can be characterized with a non-exponential utility function, so that the absolute risk aversion is not constant.
To study the case of an arbitrary utility function $u(w)$, we focus on $d_S=2$. We consider the state and the Hamiltonian in Eq.~\eqref{eq.rhoeHdS2}, and the unitary operator
\begin{equation}
U_S = \left(
      \begin{array}{cc}
        \alpha & -\beta^* \\
        \beta & \alpha^* \\
      \end{array}
    \right)\,,
\end{equation}
with $\abs{\alpha}^2+\abs{\beta}^2=1$. Since the utility is defined up to an additive constant, without loss of generality we focus on a function $u(w)$ such that $u(0)=0$. For $q=1/2$, from Eq.~\eqref{eq.aveu} we get the general expression for the expected utility
\begin{equation}\label{eq.ud2}
\langle u(w) \rangle = \abs{\beta}^2 X -p \abs{\beta}^2 Y - 2 \text{Re}(c\alpha \beta) Z\,,
\end{equation}
where we have defined $X=u(\epsilon_2-\epsilon_1)$, $Y=2 u_o(\epsilon_2-\epsilon_1)$ and $Z=2u_o(\frac{\epsilon_2-\epsilon_1}{2})$, where $u_o(w)=(u(w)-u(-w))/2$ is the odd part of the utility function $u(w)$. We note that arbitrary $X$, $Y$ and $Z$ can be obtained from a cubic utility, for instance for $\epsilon_2-\epsilon_1=1$, we can consider
\begin{equation}
u(w) = \frac{8Z-Y}{6} w + \frac{2X-Y}{2}w^2 + \frac{2(Y-2Z)}{3} w^3\,.
\end{equation}
In detail, $Y\geq X>0$ and $Z>0$ if $u(w)$ is a  strictly increasing function.
To obtain the optimal expected utility $\mathcal U (\rho_S)$, we must to maximize Eq.~\eqref{eq.ud2} over all the complex $\alpha=a e^{i\theta}$ and $\beta=b e^{i\phi}$ such that $a^2+b^2=1$.
We search the stationary point of the Lagrangian $L=\langle u(w) \rangle-\lambda(a^2+b^2)$, where we have introduced the Lagrange multiplier $\lambda$ such that $a^2+b^2=1$. For $c$ real, we get
\begin{eqnarray}
&&\cos(\theta+\phi) = \pm 1\,,\\
&&a= \mp \frac{cZb}{\lambda}\,,\\
&&b= \frac{1}{\sqrt{1+\frac{c^2Z^2}{\lambda^2}}}\,,\\
\label{eq.lambda}&&\lambda^2 - (X-pY)\lambda - c^2Z^2=0\,,
\end{eqnarray}
where the sign $\cos(\theta+\phi)$ is chosen such that $a\geq 0$, and
\begin{equation}\label{eq.utid2}
\mathcal U (\rho_S) = \frac{X-pY+\frac{2c^2Z^2}{\lambda}}{1+\frac{c^2Z^2}{\lambda^2}}=\frac{\lambda}{1+\frac{c^2Z^2}{\lambda^2}}\,.
\end{equation}
Thus we deduce that $\lambda\geq 0$ is the largest solution of Eq.~\eqref{eq.lambda}, which explicitly reads
\begin{equation}
\lambda= \frac{X-pY}{2}+ \sqrt{\left(\frac{X-pY}{2}\right)^2+c^2Z^2}\,,
\end{equation}
so that $\mathcal U(\rho_S)\geq 0$.
We note that the expression in Eq.~\eqref{eq.utid2} can be easily generalized for a complex coherence $c$ by replacing $c^2$ with $|c|^2$.
As $|c|Z\to 0$, we get the expansion
\begin{equation}\label{eq.expaZ}
\mathcal U (\rho_S) = \frac{1}{2}(X-p Y + |X-p Y|) + \frac{|c|^4 Z^4}{|X-pY|^3} + O(|c|^6Z^6)\,,
\end{equation}
so that, for an incoherent state, i.e., for $c=0$, or for a utility function such that $Z=0$, we get $\mathcal U (\rho_S)=0$ if $X-p Y\leq 0$, otherwise $\mathcal U (\rho_S)=X-p Y>0$.
We aim to characterize the states $\rho_{SA}$ such that $\delta \mathcal U(\rho_{SA})=0$. Although for an exponential utility function they are separable, for an arbitrary utility in principle they can be not. In particular, we focus on $d_A=2$ and we start to consider the states
\begin{equation}\label{eq.Xstate}
\rho_{SA}=p \ket{00}\bra{00}+(1-p)\ket{11}\bra{11}+\frac{C}{2}(\ket{00}\bra{11}+\ket{11}\bra{00})\,,
\end{equation}
where $C$, such that $0\leq C\leq 2\sqrt{p(1-p)}$, is the quantum concurrence~\cite{hill97} and quantifies the entanglement between $S$ and $A$, e.g., for $C>0$ the state $\rho_{SA}$ is entangled. By randomly generating the states $\rho_{SA}$ in Eq.~\eqref{eq.Xstate} and the values of $X$,$Y$ and $Z$ such that $X-p Y\leq 0$, we find that there are entangled states $\rho_{SA}$ showing a practically zero gain $\delta \mathcal U(\rho_{SA})\approx 0$ (see Fig.~\ref{fig:plot}).
\begin{figure}
[t!]
\centering
\includegraphics[width=0.85\columnwidth]{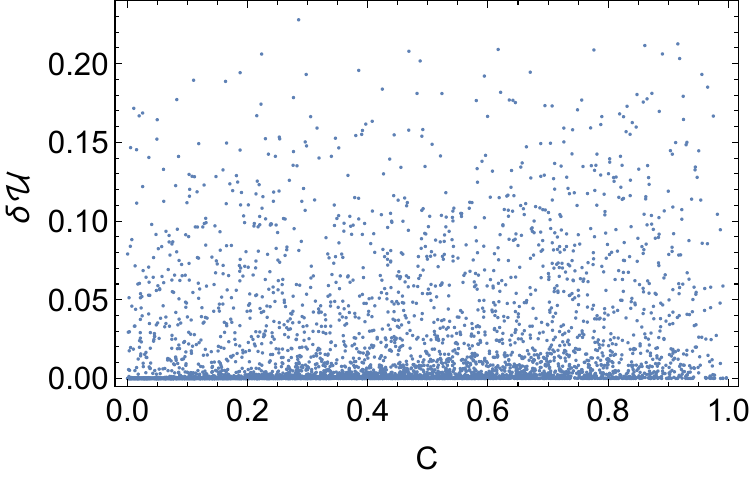}
\caption{ The gain $\delta \mathcal U(\rho_{SA})$ versus the concurrence $C$ for $10^4$ random states $\rho_{SA}$ of Eq.~\eqref{eq.Xstate}. We randomly generate $X$, $Y$ and $Z$ in the interval $[-1,1]$, such that $X-p Y\leq 0$.
}
\label{fig:plot}
\end{figure}
This suggests that there are zero-gain entangled states depending on the utility function. However, it results that all these zero-gain states are obtained when $X-q Y\leq 0$ for any $q\in [0,1]$, as $Z\to 0$, since in this case $\mathcal U(\rho_{S|a})\sim Z^4 \to 0$ for any state (see Eq.~\eqref{eq.expaZ}).
Thus, $\delta\mathcal U(\rho_{SA})=0$ for all the states $\rho_{SA}$ when $Z=0$ if $X\leq 0$ and $Y\geq 0$ or if $X\leq Y<0$. In particular, in this case the utility function is not a strictly increasing function.
Thus, for instance we have zero-gain for $X<0$ and $Y=Z=0$, i.e., for the quadratic utility $u(w)=-w^2$, or for $Y>0$ and $X=Z=0$, i.e., for $u(w) = -w -3 w^2 + 4 w^3$.
Furthermore, we note that $\delta\mathcal U(\rho_{SA})=0$ for all the states $\rho_{SA}$ when $Z=0$ if $X>Y\geq 0$ or if $X>0$ and $Y<0$, since in this case $X-q Y> 0$ for any $q\in [0,1]$ and by using Eq.~\eqref{eq.expaZ} we get $\mathcal U_{\{\Pi^A_a\}}(\rho_{SA}) = \mathcal U(\rho_S) =X-pY$ for any set $\{\Pi^A_a\}$.
We note that for utility functions $u(w)$ such that $Z=0$, the optimal expected utility is equal to
\begin{equation}\label{eq.incouti}
\mathcal U (\rho_S) = \mathcal U (\Delta(\rho_S))
\end{equation}
for any state $\rho_S$ and for the dephasing map with energy projectors $\Delta(\rho_S) =\sum_k \ket{\epsilon_k}\bra{\epsilon_k}\rho_S \ket{\epsilon_k}\bra{\epsilon_k}$. Thus, the initial quantum coherence does not give any contribution for these utility functions, which we call `incoherent' utility functions. In detail, we define the incoherent utility functions as the utility functions $u(w)$ such that $\mathcal U_c(\rho_S)=0$ for all $\rho_S$, where in general the coherent contribution is defined as $\mathcal U_c(\rho_S)= \mathcal U (\rho_S) - \mathcal U (\Delta(\rho_S)) $~\cite{Francica24}. For $d_S=2$, the incoherent utility functions are the functions $u(w)$ such that $Z=0$.
For an arbitrary incoherent utility function, the gain can be expressed as
\begin{equation}\label{eq.gainZ0}
\delta\mathcal U(\rho_{SA}) = \max_{\{\Pi^A_a\}}\sum_a p_a \mathcal U(\Delta(\rho_{S|a})) - \mathcal U(\Delta(\rho_S))
\end{equation}
and, for instance, it is zero for $d_S=2$ for any state if $X\leq 0$ and $Y\geq 0$, if $X\leq Y<0$, if $X>Y\geq 0$ or if $X>0$ and $Y<0$. For a Werner state $\rho_{SA}=\frac{1-z}{4}I + z \ket{\psi}\bra{\psi}$, where $\ket{\psi}=(\ket{00}+\ket{11})/\sqrt{2}$, which is entangled for $z>1/3$, given an incoherent utility function, from Eq.~\eqref{eq.gainZ0} we get
\begin{equation}\label{eq.werneruti}
\delta \mathcal U(\rho_{SA})=\max_{q\in[0,1]} \frac{1}{4}\abs{\tilde X+ \tilde Y_q}+ \frac{1}{4}\abs{\tilde X-\tilde Y_q}- \frac{1}{2}\abs{\tilde X}\,,
\end{equation}
where $\tilde X=X-Y/2$ and $\tilde Y_q=\left(q-\frac{1}{2}\right)z Y$. If $2X>Y\geq0$, from Eq.~\eqref{eq.werneruti} we get $\delta\mathcal U(\rho_{SA})=0$ for $z\leq z_0 = 2 X/Y-1$, whereas  $\delta\mathcal U(\rho_{SA})>0$  for $z>z_0$. Thus, for $z_0>1/3$ there are Werner states with zero-gain (for $z\leq z_0$) although they are entangled states, and there are also entangled states with nonzero-gain, e.g., Werner states with $z>z_0$. The results can be summarized as in Table~\ref{table:1}.
\begin{table}
[t!]
\caption{The role of quantum correlations can be summarized as follows.}
\vspace{0.2cm}
\begin{tabular}{|c|c|}
\hline
risk aversion & utility gain\\
\hline
\hline
$r_A=0$ &$\delta \mathcal U =\delta \mathcal E =0 \Rightarrow \rho_{SA}\in \mathcal C_S$  \\
\hline
$r_A=\text{const}\neq 0$ & $\delta \mathcal U =0 \Rightarrow\rho_{SA}\in \mathcal S $   \\
\hline
$r_A\neq \text{const}$& $\delta \mathcal U =0 \nRightarrow \rho_{SA}\in \mathcal S $\\
\hline
\end{tabular}
\label{table:1}
\end{table}
In contrast, in the case of arbitrary non-incoherent utility functions, if Eq.~\eqref{eq.gainZ0} is zero, the gain is given by the coherent contribution and reads $\delta\mathcal U(\rho_{SA}) = \max_{\{\Pi^A_a\}}\sum_a p_a \mathcal U_c(\rho_{S|a}) - \mathcal U_c(\rho_S)$. In this case, we find that the gain $\delta \mathcal U(\rho_{SA})$ is positive if $\rho_{SA}$ is entangled at least for $d_S=2$.
Furthermore, we note that the condition $\Delta(\rho_{S|a})=\rho_{S|a}$ for all $a$ and any set $\{\Pi^A_a\}$, from which $\rho_S=\Delta(\rho_S)$ and Eq.~\eqref{eq.gainZ0} follows, is not a sufficient condition in order to can obtain $\delta\mathcal U(\rho_{SA})=0$ for $\rho_{SA}\notin \mathcal S$.
To prove it, it is enough to note that if this condition is satisfied, by using Lemma~\ref{lemma2}, we get $\rho_{SA}\in \mathcal C_S$ having the form $\rho_{SA} = \sum_k \ket{\epsilon_k}\bra{\epsilon_k}\otimes C^A_k$. Then, in this case there are no states $\rho_{SA}\notin \mathcal S$ such that $\delta \mathcal U(\rho_{SA})=0$, for any utility function $u(w)$. This suggests that only for incoherent utility functions, which we recall to be defined such that Eq.~\eqref{eq.incouti} is satisfied for any $\rho_S$, there can be states $\rho_{SA}\notin \mathcal S$ such that $\delta \mathcal U(\rho_{SA})=0$.

\subsection{Generalization to arbitrary $q$}
Although we only focused on $q=1/2$, all the results achieved can be easily generalized to arbitrary values of the quasiprobability parameter $q$. For a constant absolute risk aversion $r_A=r$, as noted in Ref.~\cite{Francica24}, given a state $\rho_S$, the optimal value in Eq.~\eqref{eq.exputiexpo} is obtained only for $q=1/2$. Conversely, for $q\neq 1/2$ we get a larger optimal value $\mathcal U(\rho_S)$, which involves a particular affine combination of the permutations of the eigenvalues $u_k$ of the operator $e^{-r H_S/2} \rho_S e^{-rH_S/2}$. To generalize the results to an arbitrary $q$, it is useful to define the tilde map $\rho_S \mapsto \tilde \rho_S$ such that
\begin{equation}\label{eq.tildemap}
\tilde \rho_S = \frac{1}{2} \left( e^{-rq H_S} \rho_S e^{-r(1-q)H_S} + e^{-r(1-q)H_S} \rho_S e^{-rq H_S}\right)\,.
\end{equation}
Then, since
\begin{equation}
\langle u(w)\rangle = \frac{1}{r}\left(1-\Tr{U_S \tilde \rho_S U^\dagger_S e^{rH_S}}\right)\,,
\end{equation}
the optimal expected utility $\mathcal U(\rho_S)$ is still given by Eq.~\eqref{eq.exputiexpo} with new $u_k$ and $\ket{u_k}$ that depend on $q$, which are eigenvalues and eigenvectors of $\tilde \rho_S$, i.e., such that $\tilde \rho_S \ket{u_k}= u_k \ket{u_k}$ and $u_k\geq u_{k+1}$.
In particular, from Ref.~\cite{Francica24}, we deduce that the new $u_k$ (achieved for an arbitrary $q$) can be expressed as an affine combination of the permutations of the $u_k$'s for $q=1/2$, so that the minimum of $\mathcal U(\rho_S)$ over $q$ is obtained at $q=1/2$ (at least in a neighborhood of $q=1/2$).
Thus, since the tilde map defined by Eq.~\eqref{eq.tildemap} is linear, Lemma~\ref{lemma1} still holds with this new $\tilde \rho_{S|a}$ obtained by applying the tilde map to $\rho_{S|a}$. Only Theorem~\ref{theorem} undergoes a slight change in form due to the inverse of the tilde map. By solving the operator equation $a x + x a= y$ with $a\geq 0$ (see, e.g., Ref.~\cite{bhatia97}), we get the inverse
\begin{equation}
\rho_S = 2 \int^\infty_0 e^{-s A +rqH_S} \tilde \rho_S e^{-s A +rqH_S} ds\,,
\end{equation}
where $A = e^{-r(1-2q)H_S}$. Then, for arbitrary $q$ the zero-gain state in Eq.~\eqref{eq.statetheo} reads
\begin{equation}
\rho_{SA} = \sum_{k} \left(2 \int^\infty_0 e^{-s A +rqH_S}\ket{u_k}\bra{u_k}e^{-s A +rqH_S} ds\right) \otimes C^A_k\,,
\end{equation}
and Theorem~\ref{theorem} still holds with this new state. On the other hand, for a non-constant absolute risk aversion and an arbitrary $q$, Eq.~\eqref{eq.ud2} still holds with a new $Z$ that is
\begin{equation}
Z = u_o (q(\epsilon_2-\epsilon_1)) + u_o((1-q)(\epsilon_2-\epsilon_1))\,.
\end{equation}
Thus, all the results can be easily generalized to arbitrary $q$. In particular, for $q=0,1$, we get $Z=Y/2$, then for incoherent utility functions we get $Z=0$ and thus $Y=0$. In this case, the utility function is trivial (being nonzero on the support only for $w=\epsilon_2-\epsilon_1$) and gives $\mathcal U(\rho_S)=(X+|X|)/2$, which does not depend on $\rho_S$, so that $\delta \mathcal U(\rho_{SA})=0$ for any state $\rho_{SA}$.

\section{Conclusions}\label{sec.conclusions}
We considered a bipartite quantum system and took into account an agent that can extract work locally through a daemonic protocol introduced in Ref.~\cite{Francica17}, which is affected by work fluctuations when the agent is non-neutral to risk. We introduced the daemonic expected utility, so that the agent can select the optimal work extraction protocol by looking on the utility gain defined from this quantity. We completely characterized the role of correlations among the two parties of the system, showing how quantum correlations can influence the selection done by the agent depending on the absolute risk aversion.
Furthermore, our results clarify the role of initial quantum coherence with respect to the energy basis, which is tricky. We showed how the presence of an initial quantum coherence can give zero utility gain although the states have positive ergotropy gain. Furthermore, in general, if the contribution of this initial quantum coherence is always absent due to the particular form of the utility function, then there are entangled state with zero utility gain.
In conclusion, we believe our results represent a substantial step forward in understanding work fluctuations, showing how these strongly alter the optimization of the work extraction daemonic protocol (when they are taken in account by a utility function), resulting in a sensible deviation of a risk non-neutral agent's decision-making behavior from a risk neutral one.
We therefore hope that our studies can find some applications in the implementation of work extraction protocols, based on measurement and feedback, which in general can be optimized considering a utility function.

\subsection*{Acknowledgements}
The authors acknowledge financial support from the project BIRD 2021 "Correlations, dynamics and topology in long-range quantum systems" of the Department of Physics and Astronomy, University of Padova, from the European Union - Next Generation EU within the National Center for HPC, Big Data and Quantum Computing (Project No. CN00000013, CN1 - Spoke 10 Quantum Computing) and  from the Project "Frontiere Quantistiche" (Dipartimenti di Eccellenza) of the Italian Ministry for Universities and Research.

\end{document}